\documentclass[oribibl]{llncs}
\usepackage[utf8]{inputenc}
\usepackage{mathtools}
\usepackage{xspace}
\usepackage{color}
\usepackage[colorlinks]{hyperref}
\usepackage{float}
\floatstyle{boxed}
\restylefloat{figure}
\usepackage{todonotes}
\usepackage{caption}
\usepackage{subcaption}

\newcommand{\gray}[1]{\textcolor{gray}{#1}}

\usepackage{bussproofs}
\EnableBpAbbreviations

\newcommand{\tabinfer}[3][]{
\AxiomC{#2}\RightLabel{#1}
\UnaryInfC{#3}\DisplayProof
}
\newcommand{\UICm}[1]{\UIC{$#1$}}
\newcommand{\AXCm}[1]{\AXC{$#1$}}

\newcommand{\UICg}[1]{\UIC{\gray{$#1$}}}

\newcommand{\RLg}[1]{\RL{\gray{#1}}}
\newcommand{\RLm}[1]{\RL{$#1$}}
\newcommand{\RLmg}[1]{\RLg{$#1$}}

\newcommand{\RAXm}[1]{\AXC{}\RL{ax}\UICm{#1}}

\newcommand{\tableau}[1]{\rootAtTop #1  \DP}
\newcommand{\sequent}[1]{\rootAtBottom #1 \DP}

\newcommand{\mysec}[1]{Section~#1}
\newcommand{\mydef}[1]{Definition~#1}

\newcommand{\myfig}[1]{Figure~#1}
\newcommand{\mythm}[1]{Theorem~#1}

\def\fa{\forall}
\def\ex{\exists}

\def\E{\cal E}

\def\brsep{\; |\;}

\def\T{\mathcal{T}}

\def\calL{\mathcal{B}}

\def\eunif{\approx_{\E}}

\def\tclose{\odot}

\def\vd{\vdash}

\def\point{\cdot}

\def\C{\mathcal{C}}

\def\calT{{\cal T}}

\def\gs3{{\sf GS3}\xspace} 

\def\branch{\mathpzc{b}}
\def\brule{\mathpzc{r}}

\def\limp{\Rightarrow}
\def\land{\wedge}
\def\lor{\vee}
\def\lfa{\forall}
\def\fa{\lfa}
\def\lex{\exists}
\def\ex{\lex}

\def\eps{{\varepsilon}}

\title{A syntactic soundness proof for free-variable tableaux with on-the-fly
  Skolemization }
\author{Richard Bonichon\inst{1} \and Olivier
  Hermant\inst{2}}
\institute{
  Universidade Federal do Rio Grande do Norte, Natal, RN, Brazil \\
  \email{richard@dimap.ufrn.br}
 \and MINES ParisTech, PSL Research University, France \\
 \email{olivier.hermant@mines-paristech.fr}}

\DeclareMathAlphabet{\mathpzc}{OT1}{pzc}{m}{it}
\begin{document}
\maketitle

\begin{abstract}
We prove the syntactic soundness of classical tableaux with
free variables and on-the-fly Skolemization. Soundness proofs are usually
built from semantic arguments, and this is to our knowledge, the first proof that
appeals to syntactic means. We actually prove the soundness property with
respect to cut-free sequent calculus. This requires great care because of the
additional liberty in freshness checking allowed by the use of Skolem
terms. In contrast to semantic soundness, we gain the possibility to
state a cut elimination theorem for sequent calculus, under the
proviso that completeness of the method holds. We believe that such
techniques can be applied to tableaux in other logics as well.
\end{abstract}

\section{Introduction}
\label{sec:introduction}

Tableaux methods form a successful sub-family of automated theorem proving,
encompassing classical as well as modal logics. Their origin comes from Beth's
semantic considerations \cite{Beth1955}. With Smullyan's updated tree-based
formalism \cite{RSmu68}, as well as Fitting's subsequent treatment
\cite{MFit96}, there is a first separation between syntactic and semantic
concerns.
Both present a purely syntactic operational behavior of
tableaux rules, justified by semantic soundness and completeness
proofs. Proving these two properties by semantic arguments has stayed the norm
and for good reasons: model-theoretic proofs are reasonably short, relatively
elegant and straightforward. In comparison, syntactic proofs can be messy, as
all translation details must be shown.

There might be another reason. Translating ground tableaux proofs à la
Smullyan to ground sequent calculus proofs is indeed trivial. If we
allow free variables and Skolemization, we still have a
straightforward translation to Antonsen and Waaler's free-variable
sequent calculus \cite{RAntAWaa03}. Thus, the relation between
classical tableaux and sequent calculi has been relegated to folklore
knowledge.

Nonetheless, translating free variable tableaux with Skolemization to
ground sequent calculus is not as simple a task: most of the trouble
comes from the freshness conditions imposed on existential witnesses
in sequents. Despite our efforts, we were not able to find any result
on that matter.

However, why would one want syntactic soundness over semantic
soundness ? At the proof-theoretical level, it provides a
double-check of soundness. In practice, it does not add any
power to tableaux heuristics. However, it presents some benefits,
especially in the context of proof production and proof theory.

Since it is not hard to encode ground sequent calculus rules into any
proof assistant such as Coq, Isabelle or Dedukti
\cite{MBoeQCarOHer12}, if we are able to 
reconstruct a ground sequent derivation from a free-variable
Skolemized tableaux procedure, we will get (almost) free external
verification tools. On the tableaux side, a syntactic soundness proof
highlights where and how non-elementary speedups are achieved from
the use of efficient $\delta$-rules. Lastly, our long-term goal is to
derive cut elimination theorems from tableaux completeness proofs, in
extensions of first-order logic, and this requires syntactic,
cut-free, soundness proofs.


\section{Free-Variable Tableaux}
\label{sec:free-vari-tabl}

The language is usual first-order logic with predicate and
function symbols. Sets and multisets of formulas are denoted by
capital greek letters ($\Gamma$, $\Delta$), while formulas 
are denoted by upper case letters $A, B, C, D$. 
We use the lower case letters $f, g$ to denote function symbols and
$a, b, c, d$ for constants. Variables are denoted as $x, y, z$. We
also use indexes or quotes when we need more symbols.

We present tableaux as a refutation calculus with attached
constraints via a global constraint store. This global store
represents the necessary unification steps to be performed and
satisfied in order to close the tableau. A \emph{constrained tableau}
is a pair $\T\point\C$ where $\T$ is a tableau and $\C$ a set of
unification constraints.


A branch can be \emph{closed} when it carries two opposite unifiable
formulas. \emph{Unifiable} here means that the global store does not
become inconsistent when adding the new unification constraints. A
tableau is itself said closed when \emph{all its branches} can be
closed at once. In this case, all closing constraints are unifiable.

This means that closing a first-order tableau can be seen as
providing a unifier that simultaneously satisfies all the global
constraints \emph{and} the closing constraints of the open 
branches, or, equivalently, that does not induce any new constraint on
the latter branches. The constraint store keeps the minimal
requirements for such a unifier: they come from the early closure of
some branches, discussed before. Of course, if this is done
carelessly, we can come to a dead-end.

We see constraints as a degree of liberty for tableaux. Ultimately, we
just can decide not to generate constraints at all, until a global
unifier can be found. The soundness proof of
\mysec{\ref{sec:soundness}} promotes this point of view: it assumes a
unifier and no constraints.

The rules, presented in \myfig{\ref{fig:tab-rules}} where the
constraints are omitted if they are unchanged, are an extension of
usual non-destructive free-variable tableaux
calculi. Non-destructivity is not strictly needed neither for
soundness nor for completeness, but it eases some developments.

Tableaux rules are usually divided into 4 sets: 2 sets decompose logical
connectives ($\alpha, \beta$), two act on quantifiers ($\delta, \gamma$). We
need only add the closure rule ($\odot$). If
$\land, \lor, \Rightarrow, \neg, \forall, \exists$ are allowed, we have the
following groups:

\[
\begin{array}{@{\vline\hspace{.2cm}}l@{\hspace{.2cm}\vline\hspace{.2cm}}l
    @{\hspace{.2cm}\vline}}
  \hline
  \alpha & A\land B, \neg (A\lor B), \neg (A\Rightarrow B), \neg \neg
  A \\ \hline
  \beta &  A\lor B, A\Rightarrow B, \neg (A\land B) \\ \hline
  \delta & \ex x\; A, \neg (\fa x\; A) \\ \hline
  \gamma & \fa x\; A, \neg(\ex x\; A) \\ \hline
\end{array}
\]

The decomposition of formulas happens as follows: the tableaux method
matches the active formula with one of the above categories, then applies the
corresponding rule to it. Negated formulas are actually handled in two steps:
the negation is pushed to the direct subformulas, transforming the
active connective by De Morgan laws, then the decomposition of the connective is
applied.

In pure automated deduction mode, it is enough to keep only the
current set of open branches, since the rules apply only on them. This
is no more the case if we are interested in exporting the proof in
other formats \cite{RBonDDelDDol07}. Moreover, keeping track of
previous steps can help us during proof search.

For proof-theoretic purposes, it is convenient to record all the steps
of the proof and to consider a tableau derivation as a {\em tree}
rooted at the original multiset of formulas; tableau
branches are {\em nodes}, internal if they already have been applied
some rule and external (\emph{leaves}) otherwise; the leaves that are
not closed, are \emph{open}, and they constitute the tableau properly
speaking. Tableau rules primarily operate on those leaves, extending
one of them at a time: rules are recorded as labels of inner
nodes. Trees themselves enjoy a notion of branch, that we
replace, to prevent confusion, with tableaux branches, by the word {\em
  path}.

Due to the non-destructive nature of the rules, the formulas on a path
are collected at the leaves. Paths, as well as leaves/branches, will
be identified as usual with trees, with sequences of $0$ and
$1$. $b.0$ is the left child of a path $b$, (or the unique child if
there is no branching), and $b.1$ is its right child.

\begin{figure}[htbp]
{\small
  \begin{subfigure}[t]{0.2\textwidth}
      \tabinfer{$\alpha(A,B)$} {$A,B$}{$\alpha$}
  \end{subfigure}
  \begin{subfigure}[t]{0.2\textwidth}
    \tabinfer{$\beta(A,B)$}{$A \brsep B$} {$\beta$}
  \end{subfigure}
  \begin{subfigure}[t]{0.3\textwidth}
    \tabinfer{$\gamma(x,A)$} {$A(x:=X)$}{$\gamma$} \hfill\\
    $X$ fresh free variable  \hfill
  \end{subfigure}
  \begin{subfigure}[t]{0.25\textwidth}
    \tabinfer{$\delta(x,A)$}{$A[$x:=\texttt{sko}$(\text{args})]$}{$\delta$}
  \end{subfigure}  \hfill

  \begin{subfigure}[t]{0.5\textwidth}
    \tabinfer{$A,\neg A\cdot \C$}
    {$\tclose \cdot (\C\cup \{A\eunif A\})$}
    {\texttt{closure} ($\tclose$)}
  \end{subfigure}\hfill
  \begin{subfigure}[t]{0.4\textwidth}
    Constraints ($\C$) are omitted in $\alpha$, $\beta$, $\gamma$, $\delta$.\hfill
  \end{subfigure}
  \hfill
}
  \caption{Tableau expansion and closure rules.}
  \label{fig:tab-rules}
\end{figure}

$\alpha$-rules and $\beta$-rules correspond to the standard ones as
found in Smullyan's textbook \cite{RSmu68}.
They all include negated formulas, as $\neg$ is a primitive
connective, and not an operator transforming formulas into negation
normal forms.

Free variables are used in $\gamma$-rules as placeholders waiting some
satisfying term instantiation, usually given by closure. This has a
direct effect on the treatment of existential quantifiers as we now
must use Skolemization to get a suitable sound witness.

The $\delta$-rule shown is generic and produces a fresh Skolem symbol
on-the-fly. This function symbol, here named \texttt{sko}, receives
the free variables in $A$ as arguments (args). The term is therefore
guaranteed to be fresh. We use a standard \emph{inner Skolemization}
\cite{ANonCWei01}: the arguments of the Skolem symbol are the free
variables actually occurring in the Skolemized formula $A$. Inner
Skolemization is more efficient than outer Skolemization in the sense
that it uses only relevant (i.e. fewer) elements as arguments. Such
on-the-fly Skolemization can also be replaced by a
pre-inner-Skolemization of formulas (this is the $\delta^{+^{+}}$ rule
of \cite{Beckert1993}), which would be even more efficient on some
problems. We chose not to do so because we intend to extend this work to
Deduction modulo \cite{DBLP:journals/jar/DowekHK03}, which does not behaves well with
pre-Skolemization, unless we switch to polarized Deduction modulo \cite{GDow10}.

Finally, we also have chosen inner Skolemization over other forms of
strong quantifier treatments \cite{Cantone2000,  Giese1999}\cite{DBLP:journals/jar/CantoneA07} because it
adds less noise (through technical difficulties) to the syntactic soundness
proof of \mysec{\ref{sec:soundness}}.

All in all, inner Skolemization is a good tradeoff between efficiency
and simplicity. It allows us to expose the techniques that allow us to
show syntactic soundness, with the right degree of difficulty.

Let us prove Smullyan's drinker problem, $\ex{x} (D(x) \limp \fa{y}
D(y))$, where $D$ is a unary predicate. As usual with tableaux, we actually
refute the negation $\neg (\ex{x} (D(x) \limp \fa{y} D(y)))$. The full
derivation is shown in \myfig{\ref{fig:min-delta-example}}.

\begin{figure}[htbp]
\begin{center}
\tableau{
\AXCm{\odot~~\{ X \approx c \}}
\RLm{\odot}
\UICm{\neg (\ex{x} (D(x) \limp \fa{y} D(y))), \neg (D(X) \limp \fa{y}
  D(y)), D(X), \neg \fa{y} D(y), \neg D(c)}
\RLm{\delta}
\UICm{\neg (\ex{x} (D(x) \limp \fa{y} D(y))), \neg (D(X) \limp \fa{y}
  D(y)), D(X), \neg \fa{y} D(y)}
\RLm{\alpha}
\UICm{\neg (\ex{x} (D(x) \limp \fa{y} D(y))), \neg (D(X) \limp \fa{y} D(y))}
\RLm{\gamma}
\UICm{\neg (\ex{x} (D(x) \limp \fa{y} D(y)))}
}
\end{center}
\caption{A proof of the drinker principle}
\label{fig:min-delta-example}
\end{figure}



\section{Sequent Calculus} \label{sec:gs}

This section presents the sequent calculus which will be used for the
syntactic soundness proof for tableaux. This version is as
close as possible to tableaux and equivalent to more usual sequent
calculi. The important difference with tableaux is that, as
most sequent calculi\footnote{one exception is Waaler and Antonsen's
  free-variable sequent calculus\cite{RAntAWaa03}}, we do not allow free
variables nor Skolemization, which will be the major concern of
\mysec{\ref{sec:soundness}}).

\gs3\footnote{We follow Troelstra and Schwichtenberg's classification and
  naming \cite{HSchASTro00}} (for Gentzen-Sch\"utte) is a one-sided variant of Gentzen's
original {\sf LK} sequent calculus. Contraction is implicit, built
into each inference rule, both to stick to tableaux rules, 
and as a convenience for the proofs we will develop. In contrast, the
weakening rule is explicit. The cut rule is absent, as we intend to go
without it in the soundness proof. To underline the similarities with
tableaux, we split the presentation of the rules along the $\alpha,
\beta, \gamma, \delta$ (\myfig{\ref{fig:gs3}}) classification for
tableaux, except that we explicitly mention every case, which is more
customary in sequent calculi.

\begin{figure}
  \begin{minipage}[t]{0.36\textwidth}
\small
    \begin{tabular}[t]{c}
      {\bf $\alpha$ group} \\[0.2cm]
    \AXCm{\Delta, \neg\neg A, A \vd}
    \RL{$\neg\neg$}
    \UICm{\Delta, \neg\neg A \vd}
    \DP \\\\
    \AXCm{\Delta, \neg\neg A, A \vd}
    \RL{$\neg\neg$}
    \UICm{\Delta, \neg\neg A \vd}
    \DP
    \\\\
    \AXCm{\Delta, \neg (A \Rightarrow B), A, \neg B \vd}
    \RL{$\neg\Rightarrow$}
    \UICm{\Delta, \neg (A \Rightarrow B) \vd} \DP
      \\\\
    \AXCm{\Delta, A \land B, A, B \vd} \RL{$\land$}
    \UICm{\Delta,A\land B \vd} \DP \\\\
    \AXCm{\Delta, \neg (A \lor B), \neg A, \neg B \vd}
    \RL{$\neg\lor$}
    \UICm{\Delta, \neg (A \lor B) \vd}\DP \\\\
      {\bf axiom rule} \\[0.2cm]
      \AXCm{~} \RL{ax}
        \UICm{\Delta,A, \neg A\vd} \DP
    \end{tabular}
  \end{minipage}
  \begin{minipage}[t]{0.3\textwidth}
  \begin{tabular}[t]{c}
    {\bf $\beta$ group}
     \\[0.2cm]
    \AXC{
    \begin{tabular}[t]{c}
      $\Delta, A\Rightarrow B,\neg A \vd$ \\ $\Delta, A\Rightarrow B, B \vd$
    \end{tabular}}
    \RL{$\Rightarrow$}
    \UICm{\Delta, A\Rightarrow B \vd} \DP
    \\\\

      \AXC{
    \begin{tabular}[t]{c}
      $\Delta, \neg (A \land B), \neg A, \vd$ \\
      $\Delta, \neg (A \land B), \neg B, \vd$
    \end{tabular}}
      \RL{$\neg\land$}
    \UICm{\Delta,\neg (A \land B) \vd} \DP
    \\\\
    \AXCm{
    \begin{array}[t]{c}
      \Delta, A\lor B, A\vd \\
      \Delta, A\lor B, B\vd
    \end{array}}

    \RL{$\lor$}
    \UICm{\Delta, A\lor B\vd}\DP \\\\

      {\bf structural group}\\[0.2cm]
       \AXCm{\Delta \vd}
      \RL{w}
      \UICm{\Delta, A \vd} \DP

  \end{tabular}
  \end{minipage}
  \begin{minipage}[t]{0.3\textwidth}
    \begin{tabular}[t]{c}
      {\bf $\delta$ group}\\[0.2cm]
      \AXCm{\Delta, \ex x\; A(x), A(c) \vd}
    \RL{$\ex$}
    \UICm{\Delta, \ex x\; A(x)\vd} \DP \\\\
      \AXCm{\Delta, \neg \fa x\; A(x), \neg A(c)\vd} \RL{$\neg\fa$}
    \UICm{\Delta, \neg \fa x\; A(x)\vd} \DP \\\\
      {\scriptsize where $c$ is a fresh constant} \\\\

      {\bf $\gamma$ group} \\[0.2cm]
 \AXCm{\Delta, \neg \ex x\; A(x), \neg A(t) \vd}
    \RL{$\neg\ex$}
    \UICm{\Delta, \neg \ex x\; A(x)\vd} \DP \\\\
  \AXCm{\Delta, \fa x\; A(x), A(t)\vd} \RL{$\fa$}
    \UICm{\Delta, \fa x\; A(x)\vd} \DP \\\\
{\scriptsize where $t$ is any term}
    \end{tabular}
  \end{minipage}

  \caption{\gs3}
  \label{fig:gs3}
\end{figure}


\section{Soundness Proof}
\label{sec:soundness}

This section shows the following property:
\begin{theorem}[Soundness of tableaux w.r.t. \gs3]
\label{prop:tamedlksound}
Let $\Gamma$ be a set of formulas. If there is a closed tableau rooted
at $\Gamma$, with unifier $\sigma$, then the sequent $\sigma \Gamma
\vd$ has a \gs3 proof.
\end{theorem}

We require a closed tableau proof, that is to say an entire tree (see
\mysec{~\ref{sec:free-vari-tabl}}) where all branches are closed and
the constraints from the last generated constraint store
(the last rule is {\sf closure}) are satisfiable at once by some
unifier $\sigma$. It also satisfies any intermediate constraint from
this tableau proof as they all appear in the final store.

The unifier $\sigma$ can assign any term, including a free variable,
to a given free variable. To make it ground, we extend it to $\sigma'
= \kappa \circ \sigma$, where $\kappa$ maps the free variables from
the range of $\sigma$ to fresh constants. The unifier $\sigma'$ subsumes $\sigma$

Given a closed tableau proof $\T$ rooted at $\Gamma$, with ground
unifier $\sigma$, we call abusively the pair $(\T,\sigma)$
a closed tableau, which is ground and without constraint. We refer to
tableaux without unifier as \emph{strict/valid tableaux}.

\subsection{Origin of the Problem}

The naïve translation, that maps inductively each rule of $\calT$ to
the similar rule of \gs3{}, does not work. Let us translate this way
the tableau of \myfig{\ref{fig:min-delta-example}}.

The unifier is $\sigma = \{ X := c \}$, and the corresponding \gs3
pseudo-proof is the tableau proof simply turned upside down and
instantiated, as shown in \myfig{\ref{fig:unsound}} where bookkeeping
contractions have been eluded.

\begin{figure}[htbp]
\begin{center}
\sequent{
\RAXm{\neg (\ex{x} (D(x) \limp \fa{y} D(y))), \neg (D(c) \limp \fa{y}
  D(y)), D(c), \neg \fa{y} D(y), \neg D(c) \vd}
\RLm{\neg\fa}
\UICm{\neg (\ex{x} (D(x) \limp \fa{y} D(y))), \neg (D(c) \limp \fa{y}
  D(y)), D(c), \neg \fa{y} D(y) \vd}
\RLm{\neg\limp}
\UICm{\neg (\ex{x} (D(x) \limp \fa{y} D(y))), \neg (D(c) \limp \fa{y}
  D(y)) \vd}
\RLm{\neg\ex}
\UICm{\neg (\ex{x} (D(x) \limp \fa{y} D(y))) \vd}
}
\end{center}
\caption{Pseudo sequent derivation for the tableau of
  \myfig{\ref{fig:min-delta-example}}}
\label{fig:unsound}
\end{figure}

The problem in the derivation of \myfig{\ref{fig:unsound}} is that
the $\neg\fa$ rule (the counterpart of the $\delta$ rule) requires a
\emph{fresh} constant, and it cannot be $c$, as it was previously
introduced by the first $\neg\ex$ rule. In the tableau proof of
\myfig{\ref{fig:min-delta-example}}, freshness is innocently masked by
the unknown value of $X$.

The remedy, to show the drinker principle in \gs3, is well-known:
contract the goal formula, and use once to get a fresh constant $c$
with the $\neg\fa$ rule, and in a {\em
  second time} to generate the \emph{same} constant $c$ with the
$\neg\ex$ rule.

This is a one-shot particular solution, and we provide below a general jprocedure
to treat the problem: given any tableau proof, with a
\emph{relaxed} notion of freshness, we force the sequent rules to
apply in the \emph{right} order.

\subsection{Insight into the translation}
\label{sec:example}
\label{sec:insight-into-transl}

Lax freshness is sound for two reasons. First, free variable
tableaux are semantically sound. Second, we \emph{syntactically} know
it is sound through the unifier $\sigma$. The unifiability of the
constraints ensure that there is \emph{eventually} no loop. We are in
a way guaranteed that there is a right order for the instantiations.

Practice is more subtle. Indeed, any (still naïve) attempt
to order all quantifiers of the tableau by a combination of subterm
order and precedence in formula\footnote{quantifier $\mathcal{Q}_X$
  would have priority over $\mathcal{Q}_Y$ if it is higher in the same
  formula \emph{or} if the instance (by $\sigma$) of the
  metavariable/term introduced by $\mathcal{Q}_Y$ contains the Skolem
  term introduced by $\mathcal{Q}_X$.}, topologically sort them to
unravel the tableau and get the right order for rules, fails. There is
a theoretical argument: free-variable tableaux with on-the-fly
Skolemization can be non-elementarily shorter
\cite{RHahPSch94,BaazFerm95Tab,Cantone2000}
than sequent proofs, namely because of the relaxed notion of
freshness, post-checked at unification time. This appears clearly in
\myfig{\ref{fig:unsound}}: the two precedence constraints on the
$\neg\fa$ and $\neg\ex$ rules are conflicting.

The proofs of the drinker principle gives us a hint: duplication. This
removes the above theoretical barrier, as the sequent proof now grows
much bigger than the tableau proof. This also means we will make the
translated sequent grow \emph{from the root} to its axioms, ensuring
at every step soundness (the -- open -- sequent proof is \gs3{}-valid)
and progress (one tableau rule has been considered).

Let us translate the example to have a preview of what we will
do. For the sake of readability, and in analogy with the next
sections, we let $\Gamma$ be the root formula $\neg (\ex{x} (D(x)
\limp \fa{y} D(y)))$. Translating the first three rules is easy (see
\myfig{\ref{fig:initial-part}}). Next, we face the problem discussed
above and solve it in four steps:
\begin{enumerate}
\item \label{item:save} {\bf Save} the current incomplete proof-tree.
\item \label{item:weaken}
  {\bf Clean} the targeted open leaves: remove all formulas but
  $\Gamma$ and the $\delta$ formula of interest.
\item \label{item:apply} {\bf Apply} the now legal $\delta$ rule, and
  clean more (\myfig{\ref{fig:clean-and-apply}}).
\item \label{item:graft} {\bf Graft} the saved proof-tree
  \ref{item:save} to the targeted open leaves
  (\myfig{\ref{fig:graft-step}}). In fact, make the grafts of step
  \ref{item:apply} {\bf grow} following the saved proof-tree. {\bf
    Keep the Skolem formula} as an additional side formula on the
  relevant branches (in our example: on the single grafted branch).
\end{enumerate}
After those steps, we are able to translate further the tableau, in
our case, the sole axiom rule.

\def\figspace{0.4cm}
\begin{figure}[ht!]
  \centering
  \begin{subfigure}[b]{\textwidth}
    \centering
    \sequent{
      \AXCm{\Gamma, \neg (D(c) \limp \fa{y}
        D(y)), D(c), \neg \fa{y} D(y) \vd}
      \RLm{\neg\limp}
      \UICm{\Gamma, \neg (D(c) \limp \fa{y}
        D(y)) \vd}
      \RLm{\neg\ex}
      \UICm{\Gamma \vd}
    }
    \caption{First 3 steps of the translation of \myfig{\ref{fig:min-delta-example}}}
    \label{fig:initial-part}
    \vspace*{\figspace}
  \end{subfigure}

  \begin{subfigure}[b]{\textwidth}
    \centering
  \sequent{
    \AXCm{\Gamma, \neg D(c) \vd}
    \RL{w}
    \UICm{\Gamma, \neg \fa{y} D(y), \neg D(c) \vd}
    \RLm{\neg\fa}
    \UICm{\Gamma, \neg \fa{y} D(y) \vd}
    \RL{w}
    \UICg{\Gamma, \neg (D(c) \limp \fa{y}
      D(y)), D(c), \neg \fa{y} D(y) \vd}
    \RLmg{\neg\limp}
    \UICg{\Gamma, \neg (D(c) \limp \fa{y} D(y)) \vd}
    \RLmg{\neg\ex}
    \UICg{\Gamma \vd}
  }
  \caption{Cleaning and applying the $\delta$-rule}\label{fig:clean-and-apply}
  \label{fig:delta-step}

    \vspace*{\figspace}
  \end{subfigure}
  \begin{subfigure}[b]{\textwidth}
    \centering
  \sequent{
    \AXCm{\Gamma, \mathbf{\neg D(c)}, \neg (D(c) \limp \fa{y} D(y)), D(c), \neg
      \fa{y} D(y) \vd}
    \RLm{\neg\limp}
    \UICm{\Gamma, \mathbf{\neg D(c)}, \neg (D(c) \limp \fa{y} D(y)) \vd}
    \RLm{\neg\ex}
    \UICg{\Gamma, \neg D(c) \vd}
    \RLg{w}
    \UICg{\Gamma, \neg \fa{y} D(y), D(c) \vd}
    \RLmg{\neg\fa}
    \UICg{\Gamma, \neg \fa{y} D(y) \vd}
    \RLg{w}
    \UICg{\Gamma, \neg (D(c) \limp \fa{y}
      D(y)), D(c), \neg \fa{y} D(y) \vd}
    \RLmg{\neg\limp}
    \UICg{\Gamma, \neg (D(c) \limp \fa{y}
      D(y)) \vd}
    \RLmg{\neg\ex}
    \UICg{\Gamma \vd}
  }
  \caption{Graft and grow}
  \label{fig:graft-step}

  \end{subfigure}
  \caption{Solving the drinker problem}
\end{figure}

Grafting a proof-tree with more than one open leaf multiplies the
number of leaves of the tree. Translating a \emph{single} tableau
rule into \emph{several} sequent rules is unavoidable, and both height
and width grow. So, in general, a single tableau branch (resp. rule)
corresponds to several sequent branches (resp. rules). The general
mechanism is discussed in the next sections.

\subsection{Initial Definitions and Lemmas}

We have already mentioned that the \gs3 proof is not built by
structural induction. We thus need some additional definitions.

\begin{definition}[Initial part]
Let $T$ be a closed strict tableau rooted at $\Gamma$. An open tableau
$T_{0}$ is said to be an initial part of $T$ iff it is
rooted at $\Gamma$ and:
\begin{itemize}
\item either $T_{0}$ is a leaf:
  \begin{itemize}
    \item if the root of $T$ is also a leaf (closed by hypothesis), $T_{0}$ is a
      \emph{closed} leaf;
    \item if the root of $T$ is an internal node, $T_{0}$ is an
      \emph{open} leaf.
  \end{itemize}
\item or the rule applied at the root of $T_{0}$ is exactly the same as
  the rule applied at the root of $T$ and the sub-tableau(x) of $T_{0}$
  are initial parts of the corresponding sub-tableau(x) of $T$.
\end{itemize}
We use the same terminology for \gs3 proof-trees.
\end{definition}

Alternatively, if we consider a \emph{sequence of tableaux} used to
derive tableau $T$ from its root $\Gamma$, then $T_{0}$ is an initial
part of it if, and only if, there exists at least one such sequence
where $T_{0}$ appears.

An initial part $T_{0}$ of $T$ shares the same root, nodes, sequents,
branches, constraints, paths and rules as $T$ up to the leaves of
$T_{0}$. $T_{0}$ can also be thought of a labeling of the nodes of $T$ as
``seen'' and ``unseen''. For instance, the tableau of
\myfig{\ref{fig:initial-part}} is an initial part of the tableau
of \myfig{\ref{fig:min-delta-example}}.

The following lemma shows that subsequent definitions are well-formed:
\begin{lemma}
Let $T_{0}$ be an initial part of a closed strict tableau $T$, $\branch$ an
open leaf of $T_{0}$, and $\brule$ the rule applied to the corresponding
branch $\branch$ on $T$. The extension of $T_{0}$ by the application of $\brule$ on
$\branch$ is also an initial part of $T$.
\end{lemma}

Our goal is to incrementally build a \gs3{} proof-tree by following the
rules of $T$,  given a closed (strict) tableau $T$ with a ground unifier
$\sigma$. In a sense, we replay the
steps that were used to build $T$, get an initial part $T_{0}$, and
maintain the invariant that the \gs3{} proof-tree \emph{maps to}
$T_{0}$. Note again that a single open-branch of $T_{0}$ serves to extend
{\em several} branches of the \gs3{} proof-tree at the same time. We
first define the mapping:

\begin{definition}[Partial Link]\label{def:partlink}
Let $\pi_{0}$ be an open \gs3{} proof-tree rooted at $\Gamma$ and let also $s_1,
\cdots, s_n$ be its open leaves, containing respectively the sequents
$\Gamma_{s_1} \vdash, \cdots, \Gamma_{s_n} \vdash$.

Let $T_{0}$ be an open strict tableau with open leaves $\branch_1, \cdots,
\branch_m$, that respectively containing the set of formulas $\Delta_{\branch_1},
\cdots, \Delta_{\branch_m}$. Let $\sigma$ be a unifier for $T_{0}$.

$\pi_{0}$ is partially linked to $(T_{0},\sigma)$ if, and only if, there
exists a partial mapping $\mu: \{s_1,\cdots,s_n\} \mapsto
\{\branch_1,\cdots,\branch_m\}$, such that $\sigma \Delta_{\mu(s)} \subseteq
\Gamma_{s} $, when $\mu(s)$ is defined.

We say that the leaf $s$ (of $\pi_{0}$) is linked to the leaf $\mu(s)$
(of $T_{0}$), and that the formulas of $\Gamma_{s} \backslash
\sigma\Delta_{\mu(s)}$ are the {\em side formulas} of $s$.
\end{definition}

This notion is readily extended to describe a partial link to a \gs3{}
proof-tree. In this case, there is no need for an unifier.

Notice that, when $\mu(s_i) = \mu(s_j)$, nothing prevents the side
formulas of $s_i$ and $s_j$ to be different. $\Gamma_s$ is only
required to contain the instances by $\sigma$ of the formulas of
$\Delta_{\mu(s)}$.

Notice also that $\mu$ is not required to be injective or
surjective. Non-injectivity accounts for the fact that a single
tableau branch is reflected at more than one place on a \gs3{}
proof-tree. Non-surjectivity of the mapping amounts for the fact
that some branches of the original proof may not be reflected in
$\pi$, in particular when $\pi$ is bilinked (\mydef{\ref{def:bilink}}
below). One can check that, in
the proof of \mythm{\ref{thm:delta-thm-2}}, the link to $\theta$ is not
surjective, but the link to $\pi$ is maintained surjective.

We need the two following refinements over partial links:

\begin{definition}[Link]\label{def:link}
Let $\Gamma$ be a set of formula. Let $\pi_{0}$ be a proof-tree linked to
a tableau $(T_{0},\sigma)$, and assume that:
\begin{itemize}
  \item $\pi_{0}$ and $T_{0}$ are both rooted at $\Gamma$,
  \item and the mapping $\mu$ is total.
\end{itemize}
Then $\pi_{0}$ is said to be \emph{linked} to $(T_{0},\sigma)$.
\end{definition}

\begin{definition}[Bilink]\label{def:bilink}
We say that $\pi$, with open leaves $\{ s_1, \cdots, s_n \}$ is
bilinked to two \gs3{} proof-trees $\theta_{0}$ and $\theta_{1}$ if, and only
if, it is partially linked to $\theta_{0}$ and to $\theta_{1}$, and the
respective mappings $\mu_{0}$ and $\mu_{1}$ verify the disjointness and
covering conditions:
\begin{itemize}
\item $\mbox{\sf{Dom}}(\mu_{0}) \cap \mbox{\sf{Dom}}(\mu_{1}) = \emptyset$
\item $\mbox{\sf{Dom}}(\mu_{0}) \cup \mbox{\sf{Dom}}(\mu_{1}) = \{ s_1, \cdots, s_n \}$
\end{itemize}
\end{definition}

Given a link $\mu$ between a \gs3{} open proof-tree $\pi$ and an
initial part of $T$, the intention is to apply to all the open leaves
$s \in \mu^{-1}(\branch_j)$, the same rule as on $\branch_j$. This is formalized
in the next definition:

\begin{definition}[Parallel extension]
Let $\pi_{0}$ be a \gs3{} proof-tree, linked to $(T_{0},\sigma)$ with
mapping $\mu_{0}$, where $T_{0}$ is an initial part of a closed strict
tableau $T$ with unifier $\sigma$. Let $T_{1}$ be the extension of $T_{0}$
along $T$ on some open leaf $\branch$ with rule $\brule$.

The open proof-tree $\pi_{1}$ of \gs3{} is called a parallel extension
of $\pi_{0}$ along $T_{1}$ (by $\brule$) if it can be linked to $(T_{1},\sigma)$
such that the mapping $\mu_{1}$ is equal to $\mu_{0}$, except on the newly
created leaves of $\pi_{1}$, in which case the new leaves are mapped to
the corresponding premise leaf(s) of $\brule$ in $T_{1}$.

By abuse of language, this process is called the parallel extension of
$\pi$ along $T$. The equivalent notion can be defined for two
(partially) linked \gs3{} proofs-terms and we will use the same
terminology.
\end{definition}

In practice, $\pi_{1}$ is built out of $\pi_{0}$ by adding the inference
rule $\brule$ on the suitable leaves. Since this consumes exactly one rule
of $T$, the process of parallel extension eventually stops and generates a \gs3{}
proof-tree. This proof-tree is a sequent proof: all its leaves are closed because they
are totally linked to leave themselves closed. The main question is
whether this is always possible. The example in
\mysec{\ref{sec:example}} shows that it is not so simple.

\subsection{Parallel Extensions}
\label{sec:parallel-extensions}
Now we are equipped to describe our algorithm and prove the
following theorem:

\begin{theorem} \label{thm:parallel-extension}
Given any closed tableau $T$ with unifier $\sigma$, any initial part
$T_{0}$, and any \gs3{} proof-tree $\pi_{0}$ linked to $(T_{0},\sigma)$, it is
possible to parallely extend $\pi_{0}$ along $T$.
\end{theorem}
\begin{proof}
Let $\branch$ be an open leaf of $T_{0}$, and $\brule$ the rule applied to it
in $T$. Let $T_{1}$ be the extension of $T_{0}$ along $T$ on $\branch$ with rule
$\brule$. Consider the different cases for $\brule$:
\begin{itemize}
\item $\brule$ is an $\alpha$-rule on a formula $A$: on each $s_i \in
  \mu_{0}^{-1}(\branch)$, $\sigma A$ is present on $s_i$ by definition of
  linkedness, we apply $\brule$ on it. We link this new proof-tree
  exactly as the old one, and let $\mu_{1}$ be defined as:
  $$
  \left\{
  \begin{array}{c@{~=~}ll}
    \mu_{1}(s_j)   & \mu_{9}(s_j)       & \;\mbox{for any} s_j \notin \mu_{0}^{-1}(\branch) \\
    \mu_{1}(s_n.0) & \mu_{0}(\branch).0 & \;\mbox{for any} s_i \in \mu_{0}^{-1}(\branch) \\
  \end{array}
  \right.
  $$
Since both the tableau and the \gs3{} rules are non-destructive,
the invariant $\sigma \Delta_{\mu(s)} \subseteq \Gamma_{s}$ is
maintained.
\item $\brule$ is a $\gamma$-rule: we do exactly the same.
\item $\brule$ is a $\beta$-rule. We act similarly, except that we
  have two new open leaves in $T_{1}$, $\branch.0$ and $\branch.1$. As well,
  all the $s_i$ open leaves of $\pi_{0}$ split into $s_i.0$ and $s_i.1$.
  The new linking function $\mu_{1}$ is straightforward:
  $$
  \left\{
  \begin{array}{c@{~=~}ll}
    \mu_{1}(s_j)   & \mu_{0}(s_j)       & \;\mbox{for any} s_j \notin \mu_{0}^{-1}(\branch) \\
    \mu_{1}(s_n.0) & \mu_{0}(\branch).0 & \;\mbox{for any} s_i \in \mu_{0}^{-1}(\branch) \\
    \mu_{1}(s_n.1) & \mu_{0}(\branch).1 & \;\mbox{for any} s_i \in \mu_{0}^{-1}(\branch) \\
  \end{array}
  \right.
  $$
\item $\brule$ is a $\delta$-rule: this is entailed by
  \mythm{\ref{thm:delta-thm-2}} below. We postpone this case to the
  end of \mysec{\ref{sec:parallel-delta}}.
\item $\brule$ is a closure rule: we apply the axiom rule on each $s_i \in
  \mu_{0}^{-1}(\branch)$. $\branch$ is now a closed leaf of $T_{1}$, and
  accordingly
  the $s_i$ are no more open. We thus need restrict
  the domain of $\mu_{0}$: $\mu_{1} = \mu_{0_{|I}}$, where $I = \{ s_j~|~
  \mu_{0}(s_j) \neq \branch \}$.\qed
\end{itemize}
\end{proof}

Notice that the choice of the leaf $\branch$ is not imposed. In order to
optimize the translation, it is possible to define some heuristics to
choose the branch. As well, for better performances, the heuristics
may rearrange, on each path, the order of the rules but the
theoretical barrier discussed above will still pop up at some
point. This is why we do not insist on optimization here.

\subsection{Parallel $\delta$-extensions}
\label{sec:parallel-delta}

The possibility of a $\delta$-extension is made possible by the
following theorem:

\begin{theorem}[$\delta$-theorem]\label{thm:delta-thm-2}
Let $(T, \sigma)$ be a closed tableau. Let $\Gamma$ its root formulas,
$\ex{x} D(x)$ be a formula of it, on which a $\delta$-rule is applied,
generating the Skolem term $\delta$ and the formula let $D_\delta =
D(\delta)$. We consider the instances by $\sigma$ of those term and
formulas, and call them identically.

Let $\theta$ be an (open) \gs3{} proof-tree composed only with
formulas that appear in $(T, \sigma)$ (as instances by $\sigma$ of
formula of $T$), rooted at $\Gamma$ and such that  each leaf contains
at least $\Gamma$.

Assume that a set of leaves, denoted $\calL$, contains
$\ex{x} D(x)$. Let $\pi_{0}$ be an initial part of $\theta$.


Then it is possible to build a proof-tree $\pi_{1}$, rooted at
$\Gamma$, that is bilinked to $\pi_{0}$ and $\theta$ with mappings
$\mu_{\pi_{0}}$ and $\mu_\theta$ respectively, such that:
\begin{itemize}
\item There is no $s_1$ such that $\mu_\theta(s_1) \in \calL$, i.e.
  the leaves of $\theta$ in $\calL$ are ``unreachable''.
\item for any leaf $s_1$, such that $\mu_{\pi_{0}}(s_1)$ is a prefix of a
  path $\branch \in \calL$ (for short: $\mu_{\pi_{0}}(s_1)$ is a prefix of
  $\calL$), $D_\delta$ appears on this node as a side formula.
\item All other leaves $s_1$ of $\pi_{1}$ have the same formulas than
$\mu_{\pi_{0}}(s_1)$, or than $\mu_\theta(s_1)$.
\end{itemize}
\end{theorem}

\begin{proof}
We build $\pi_{1}$ by induction on the pair (size of Skolem term
$\delta$, size of $\pi_{0}$).

First of all, if $\pi_{0}$ has no rule, there is a tension between the
imposed formulas at the root of $\pi_{0}$, $\Gamma$, and the leaves of
$\pi_{1}$ linked to a prefix of $\calL$, that contain (at least) $\Gamma,
D_\delta$. That prevents $\pi_{1}$ to be $\pi_{0}$ itself. Indeed, we start
with a manipulated clone of $\theta$ and we \emph{graft} $\Gamma,
D_\delta$ at the leaves $\calL$ of $\theta$, as follows:
\begin{itemize}
\item  We let $\pi_{1}$ be $\theta$ where, to all the
  leaves $\branch \in \calL$ we have weakened to get $\Gamma, \ex{x} D$,
  applied the $\delta$-rule to generate $D_\delta$, and weakened once
  again on $\ex{x} D$. There is no freshness problem, since $\Gamma$
  does not contain any Skolem term or symbol. $\pi_{1}$ has the
  same leaves as $\theta$, except for a new set of leaves, which we
  call $\calL^\dag$. It is composed of the $\branch^\dag =
  \branch.0^{k_\branch}$, where $\branch \in \calL$ and $k_\branch$ is
  the necessary number of $0$ introduced by the $\delta$-rule and the
  weakenings. The formulas of the leaves in $\calL^\dag$ are exactly
  $\Gamma, D_\delta$.
\item We define the bilink in the following way:
  \begin{itemize}
    \item $\mu_\theta$ is the partial link from $\pi_{1}$ to
      $\theta$ defined on all the leaves $\branch$ of
      $\pi_{1}$ that are not member of $\calL^\dag$. It is merely
      the identity:
      $$
      \mu_\theta(\branch) = \branch~\mbox{if } \branch \notin \calL^\dag
      $$
    \item $\mu_{\pi_{0}}$ is the partial link from $\pi_{1}$ to
      $\pi_{0}$ defined on $\calL^\dag$. It is
      the constant $0$ function, since $\pi_{1}$ has no rule:
      $$
      \mu_{\pi_{0}}(\branch^\dag) = 0~\mbox{if } \branch^\dag \in \calL^\dag
      $$
  \end{itemize}
\end{itemize}

Otherwise, $\pi_{1}$ has at least one rule. Then, we
consider any initial part $\pi_0^{0}$ of $\pi_{0}$, that has one rule less
and is still an initial part of $\theta$. Let us call $\pi_0^{1}$ the
proof-tree produced by the induction hypothesis, with mapping
$\mu_{\pi_{0}}^0$ (resp. $\mu_{\theta}^0$) from $\pi_0^{1}$ to $\pi_0^{0}$
(resp. $\theta$).

To go from $\pi_0^{0}$ to $\pi_{0}$, a rule $\brule$ is applied on leaf
$\branch$. We have the following cases:

\begin{itemize}
\item $\branch$ is \emph{not} a prefix of $\calL$.  we simply
  copy the rule on each branch $s_0^1$ of $\pi_0^{1}$ linked to $\branch$,
  i.e. such that $\mu_{\pi_{0}}^0(s_0^1) = \branch$. The bilink is
  formed with an unchanged $\mu_\theta$. $\mu_{\pi_{0}}$ is
  straightforwardly defined from $\mu_{\pi_{0}}^0$ as in the proof of
  \mythm{\ref{thm:parallel-extension}}.

\item $\branch$ is a prefix of $\calL$ and $\brule$ is an
  $\alpha$-,$\beta$-,$\gamma$-rule: we simply copy the rule on
  each branch $s_0^1$ of $\pi^{1}_0$ linked to $\branch$, let $\mu_\theta$
  unchanged and let $\mu_{\pi_{0}}$ be defined from $\mu_{\pi_{0}}^0$ as in the
  proof of \mythm{\ref{thm:parallel-extension}}.

  In the case of a branching $\beta$-rule, we weaken on $D$ on $s_0^1.0$
  (resp. on $s_0^1.1$), if
  $\branch.0$ (resp. $\branch.1$) is no more a prefix of $\calL$. At
  least one of $\branch.0$ and $\branch.1$ is a prefix of $\calL$.

\item $\branch$ is a prefix of $\calL$, $\brule$ is a $\delta$-rule and either
  the Skolem term $\varepsilon$ is not comparable to $\delta$
  for the subterm relation, or it contains $\delta$ as a
  subterm: in this case, we copy the rule as above, since the
  Skolem term is still fresh.

\item $\branch$ is a prefix of $\calL$, $\brule$ is a $\delta$-rule and
  the Skolem term $\varepsilon$ is exactly $\delta$. Since only
  formulas of $T, \sigma$ appear, the Skolem formula must be
  exactly $D_\delta$, otherwise the term would be different. By
  induction hypothesis on $\pi^{1}_0$ and $\mu_{\pi_{0}}^0$, $\branch$ already
  contains $D_\delta$ as a side formula. $\pi^{1}_0$ has already the
  desired form and we let $\pi_{1} = \pi^{1}_0$, $\mu_{\pi_{0}} = \mu_{\pi_{0}}^{0}$ and
  $\mu_\theta = \mu_\theta^0$.

\item $\branch$ is a prefix of $\calL$, $\brule$ is a $\delta$-rule and the
  Skolem term $\eps$ is a strict subterm of $\delta$. Let $E_\eps$
  be the Skolem formula and $\ex{y} E$ the quantified formula. We
  cannot apply the $\delta$-rule on $\ex{y} E$ because $\eps$ is not
  fresh. As well, we cannot recover freshness by weakening on
  $D_\delta$, since this loses the invariant.

  But, since $\eps$ is a strict subterm of $\delta$, we can apply the
  induction hypothesis on $\pi^{1}_0$, on $\eps$ with
  the formula $\ex{y} E$, the set of leaves $\calL_{\branch} =
  {\mu_{\pi_{0}}^0}^{(-1)}(\branch)$ and with $\pi^{1}_0$ as an initial
  part of itself.

  We get a proof-tree, that we call (on purpose) $\pi_{1}$, along
  with a bilink $\mu^{1},\mu^{2}$ to $\pi^{1}_0$ and $\pi^{1}_0$. Let
  $s$ be a branch of $\pi_1$. $\mu^{2}(s) \notin \calL_\branch$, because
  ``no $\mu^{2}(s)$ can be a prefix of $\calL_\branch$'', and as we
  chose $\pi^{1}_0$ as an initial part of itself, being a prefix means
  being equal. Therefore, if $s$ is linked to a prefix of
  $\calL_\branch$, we must have $\mu^{1}(s) \in \calL_\branch$ and $s$
  contains the formulas:
  \begin{itemize}
    \item $E_\eps$ by the very hypothesis of \mythm{\ref{thm:delta-thm-2}}
    \item all the formulas of the corresponding branch of $\calL_\branch$
      by the definition of a partial link, that is to say the
      formulas of the branch $\branch$, plus the formula $D_\delta$ since
      $\branch$ is a prefix of $\calL$.
  \end{itemize}
  Therefore all those branches contain the formulas of the branch
  $\branch.0$ of $\pi_{1}$, plus the side formula $D_\delta$.

  We now proceed to the definition of the bilink of $\pi_{1}$ with $\pi_{0}$
  and $\theta$:
  \begin{itemize}
    \item $\mu_{\pi_{0}}(s) = \branch.0$ if $\mu^{1}(s)$ is defined and belongs to
      $\calL_{\branch}$, otherwise said if $\mu_{\pi_{0}}^0 (\mu^{1}(s)) = \branch$.
    \item $\mu_{\pi_{0}}(s) = \mu_{\pi_{0}}^0 ([\mu^{1} \sqcup \mu^{2}](s))$ if
      $\mu_{\pi_{0}}$ is defined on $[\mu^{1} \sqcup \mu^{2}](s)$ and different
      of $\branch$. The merge $\sqcup$ is well-defined because of the
      bilink $\mu^{1},\mu^{2}$ is disjoint.
    \item $\mu_\theta(s) = \mu_\theta^0 ([\mu^{1} \sqcup \mu^{2}](s))$
      otherwise, which is defined exactly when the two other cases
      fail.
  \end{itemize}
  We indeed compose the partial link functions, except when it
  comes to the branch $\branch$. It is easy to see that it is a bilink
  (\mydef{\ref{def:bilink}}). Moreover, let us check
  the conditions of the theorem:
  \begin{itemize}
    \item no leaf $s$ such that $\mu_\theta(s)$ is defined is a
      prefix of $\calL$ because this property holds for
      $\mu_\theta^0$. The leaves $s$ linked to a prefix of
      $\calL$ are either such that $\mu_{\pi_{0}}(s) = \branch.0$ or such that
      $\mu_{\pi_{0}}(s) = \mu_{\pi_{0}}^0 ([\mu^{1} \sqcup \mu^{2}](s))$.

    \item the leaves linked to a prefix of $\calL$  have $D_\delta$,
      and only $D_\delta$, as a side formula.

      In the case $\mu_{\pi_{0}}(s) =
      \mu_{\pi_{0}}^0 ([\mu^{1} \sqcup \mu^{2}](s))$ , this is true by
      hypothesis on $\mu_{\pi_{0}}^0$ (it adds exactly $D_\delta$ as a side
      formula) and on $\mu^{1}/\mu^{2}$, that preserve the formulas,
      since $\mu^{1}(s)$ does not belongs to/is not a prefix of (which
      is the same here) $\calL_\branch$.

      In the case $\mu_{\pi_{0}}(s) = \branch.0$, this property has
      been checked above.

    \item all other leaves have the same formulas as the branch they
      are linked to. This is an inductive property of the partial
      links $\mu_{\pi_{0}}^0$, $\mu_\theta^0$, $\mu^{1}$ and $\mu^{2}$.
  \end{itemize}

  As a remark, we can see that, if the partial links $\mu^{1}$ and
  $\mu_{\pi_{0}}^0$ are surjective, then the partial link
  $\mu_{\pi_{0}}$ is also surjective.\qed
\end{itemize}
\end{proof}

We conjecture that we can restrict ourselves, in
\mythm{\ref{thm:delta-thm-2}}, to the case of a single rule $\brule$ that
applies on all branches of $\pi_{0}$ that are a prefix of $\calL$. In this
case, we can apply $\brule$ on all the leaves that are mapped to a prefix
of $\calL$ \emph{at once}, that can save us to investigate them one by
one.

Notice that considering a \emph{set} of leaves $\calL$ is essential to
be able to apply induction hypothesis twice. This need comes from the
fact that we duplicate parts of the proof, and formulas and rules are
duplicated: a single tableau rule can be applied several times, in
parallel, in the corresponding sequent proof.

We are now in position to show the remaining case of
\mythm{\ref{thm:parallel-extension}} dealing with $\brule$ when it is a
$\delta$-rule : let $\delta$ be the Skolem term,
and $D_\delta$ the Skolem formula, after instantiation by $\sigma$. We
apply \mythm{\ref{thm:delta-thm-2}} to $\theta = \pi_{0}$, with $\calL =
\mu_{0}^{-1}(\branch)$, and $\pi_{0}$ as an initial part of $\theta$. Due to the
non-destructive nature of \gs3{}, $\Gamma$ appears on each
leaf of $\theta$. We get a proof-tree $\pi_{1}$ bilinked to $\theta/\pi_{0}$,
that is to say linked to $\pi_{0}$ by $\mu_{1} = \mu_\theta \sqcup \mu_{0}$,
where all the branches linked to $\calL$ (equivalently such that
$\mu_{0}(\mu_{1}(s)) = \branch$) contain $D_\delta$ as a side formula. $\mu_{1}$ is a link
because of the covering condition in
\mydef{\ref{def:bilink}}.

Therefore, we have a link $\mu$ from $\pi_{1}$ to
$(T_{1},\sigma)$, defined by $\mu(s) = \mu_{0}(\mu_{1}(s))$ if $\mu_{0}(\mu_{1}(s))
\neq \branch$, and $\mu(s) = \branch.0$. The parallel $\delta$-extension
has succeeded as well.\qed

Lastly, to show \mythm{\ref{prop:tamedlksound}}, we need to follow
strictly \gs3{} rules, that is to say replace the Skolem
terms by fresh constants on the proof-tree obtained by iterating
\mythm{\ref{thm:parallel-extension}}. Since Skolem term are
now \emph{fresh}, this boils down to replacing each Skolem term by a
different constant.\qed


\section{Related work and Conclusion}
\label{sec:conclusion}
\label{sec:related-work}

The effect of using optimized versions of Skolemization has been well studied
for tableaux methods on classical logic.

The increased efficiency resulting from the use of optimized Skolemization in
tableaux methods to handle existential quantifiers has seen a nice body of
work. Baaz and Fermüller \cite{BaazFerm95Tab} show how a more efficient
$\delta^{\star}$-rule, which offers non-elementary speedups in proofs. Even more
efficient $\delta$-rules, in terms of potential speedups, are presented by
Cantone and Nicolosi Asmundo \cite{Cantone2000} with the
$\delta^{\star^{\star}}$ variant
and by Giese and Ahrendt \cite{Giese1999} with the Hilbert's symbol based
$\delta^{\varepsilon}$ rule. All these enhanced Skolemization procedures are
instances of Cantone's and Nicolosi Asmundo's theoretical
framework\cite{DBLP:journals/jar/CantoneA07}. These demonstrated speedups can be
paralleled to the exponential explosion one might
experience when syntactically reconstructing tableaux proofs as ground sequent
derivations.

The technique we use in this paper to show a syntactic soundness proof
for first-order free variable classical tableaux with Skolemization
consists in linking proof-trees to synchronize their simultaneous
expansions. We are hopeful this can be extended to handle other $\delta$-variants.
The need for grafting various sub-trees during the construction of
sequent proof, to take into account the relative freshness of the
Skolem terms, and the consequent growth in width and breadth confirm
that, in presence of free variables and Skolemization, tableaux proofs
are necessarily shorter in a non-elementary way
\cite{BaazFerm95Tab}. This process can indeed make the size of the
sequent proof explodes. Our proof also confirms that semantic
arguments are shorter and often clearer, even though syntactic
transformations are needed in the context of proof verification.

It has to be noticed that (pre-) outer Skolemization or Skolemization
after a prenex normal form transformation would ease \emph{a lot} the
soundness proof. Since tableaux do not bear any $\delta$ rule, we
could translate directly the proof in \gs3{}, and apply Skolem theorem
(if $\fa{x} A(f(x)) \vd$ then $\fa{x} \ex{y} A(y) \vd$). In
particular, the proof-tree does not grow, as there is no speedup in
tableaux.

Our result is not specific to sequent calculus, it also readily
applies to turn free-variable tableaux with Skolemization into
tableaux without free variables, and should generalize gently to other
logics. In particular, our next goal is to lift this work to the
context of deduction modulo \cite{DBLP:journals/jar/DowekHK03}
, to de-Skolemize proofs, and obtain proofs checkable by tools such as
Coq or Dedukti \cite{MBoeQCarOHer12}.

The advantage of a syntactic transformation that avoids to appeal to
the cut rule, as our, is that it paves the way for a \emph{cut
  admissibility theorem}. Indeed, from a sequent calculus proof with
cuts, we would first get universal validity by (sequent) soundness,
then derivability of a tableau proof by completeness, and next, a
cut-free sequent-calculus proof by our method. Cut elimination is
known since the early days of logic for \gs3{}, this is why switching
to other calculi is interesting. In particular, in deduction modulo,
cut elimination depends on the chosen rewrite system.

We could also automate the transformation, by writing a program,
eventually certifying it in Coq, for instance through a certified
programming environment as FoCaLiZe \cite{CDubTHarVDon04}.



\bibliographystyle{splncs03}
\bibliography{biblio}
\end{document}